\documentclass[letterpaper, 11pt]{article}
\usepackage{amssymb,amsmath}
\usepackage{epsf}
\usepackage{times,bm,mathrsfs}  % Fonts
\usepackage{amsmath}
\usepackage{amssymb}
\usepackage{amsfonts}
\usepackage{mathrsfs}
\usepackage{bbm}
\usepackage{color}
\usepackage{graphicx}
\usepackage{amscd}
\usepackage{epsfig}

\definecolor{Red}{rgb}{1,0,0}
\definecolor{Blue}{rgb}{0,0,1}
\definecolor{Olive}{rgb}{0.41,0.55,0.13}
\definecolor{Green}{rgb}{0,1,0}
\definecolor{MGreen}{rgb}{0,0.8,0}
\definecolor{DGreen}{rgb}{0,0.55,0}
\definecolor{Yellow}{rgb}{1,1,0}
\definecolor{Cyan}{rgb}{0,1,1}
\definecolor{Magenta}{rgb}{1,0,1}
\definecolor{Orange}{rgb}{1,.5,0}
\definecolor{Violet}{rgb}{.5,0,.5}
\definecolor{Purple}{rgb}{.75,0,.25}
\definecolor{Brown}{rgb}{.75,.5,.25}
\definecolor{Grey}{rgb}{.5,.5,.5}
\definecolor{Black}{rgb}{0,0,0}

\def\path{{\tt path}}

\newcommand{\ecal}{\mathcal{E}}

\newcommand{\gcal}{\mathcal{G}}

\newcommand{\lcal}{\mathcal{L}}

\newcommand{\vcal}{\mathcal{V}}

\newcommand{\eps}{\varepsilon}

\newcommand{\bdm}{\begin{displaymath}}
\newcommand{\edm}{\end{displaymath}}
\newcommand{\bea}{\begin{eqnarray*}}
\newcommand{\eea}{\end{eqnarray*}}
\newcommand{\bean}{\begin{eqnarray}}
\newcommand{\eean}{\end{eqnarray}}

\newcommand{\prob}{\mathbb{P}}

\newtheorem{proposition}{Proposition}

\newenvironment{proof}{\noindent{\textbf{Proof:}}}{$\blacksquare$\vskip\belowdisplayskip}
\newcommand{\itemname}[1]{$\mathrm{[#1]}$}

\newcommand{\alleles}[2]{\mathcal{M}^{(#1)}_{#2}}
\newcommand{\loci}{\mathcal{I}}
\newcommand{\divtime}[2]{\mathcal{D}^{(#1)}_{#2}}
\newcommand{\divt}[1]{\mathcal{D}_{#1}}
\newcommand{\mrca}[1]{\langle #1 \rangle}
\newcommand{\desc}[1]{\lfloor #1 \rfloor}
\newcommand{\undert}[1]{\underline{t}_{#1}}
\newcommand{\overt}[1]{\overline{t}_{#1}}
\newcommand{\leftd}[1]{{#1}_{\swarrow}}
\newcommand{\rightd}[1]{{#1}_{\searrow}}
\newcommand{\estdivtime}[2]{\widehat{\mathcal{D}}^{(#1)}_{#2}}
\newcommand{\estdivt}[1]{\widehat{\mathcal{D}}_{#1}}
\newcommand{\mincot}{\mu}
\newcommand{\mint}{m}
\newcommand{\consrate}{\Lambda}

\begin{document}

\title{\vspace{-3cm} Incomplete Lineage Sorting: 
Consistent Phylogeny Estimation From Multiple Loci\thanks{
E.M. is supported by an Alfred Sloan fellowship in
Mathematics and by NSF grants DMS-0528488, and DMS-0548249 (CAREER) and by 
ONR grant N0014-07-1-05-06.
}
}
\author{
{\bf Elchanan Mossel}\\
       {Department of Statistics}\\
       {University of California, Berkeley}\\
       {\small \texttt{mossel@stat.berkeley.edu}}\\
       {}\\
{\bf Sebastien Roch}\\
       {Theory Group}\\
       {Microsoft Research}\\
       {\small \texttt{Sebastien.Roch@microsoft.com}}
}
\maketitle

\begin{abstract}
We introduce a simple algorithm for reconstructing phylogenies
from multiple gene trees in the presence of incomplete lineage sorting,
that is, when the topology of the gene trees may differ
from that of the species tree. We show that our technique is statistically
consistent under standard stochastic assumptions,
that is, it returns the correct tree given sufficiently many unlinked loci.
We also show that it can tolerate moderate estimation errors.
\end{abstract}

\section{Introduction}

Phylogenies---the evolutionary relationships of a group of species---are 
typically inferred from estimated genealogical histories of one or several genes 
(or \emph{gene trees})~\cite{Felsenstein:04,SempleSteel:03}.  
Yet it is well known that such gene trees may provide misleading information
about the phylogeny (or \emph{species tree}) containing them. 
Indeed, it was observed early on that 
a gene tree may be topologically 
inconsistent with its species tree, 
a phenomenon known as \emph{incomplete lineage sorting}. 
See e.g.~\cite{Maddison:97,Nichols:01,Felsenstein:04} and references therein.
Such discordance plays little role in the reconstruction of deep phylogenetic branchings
but it is critical in the study of recently diverged populations~\cite{KnowlesMaddison:02,
HeyMachado:03,Knowles:04}.  

Two common approaches to deal with this issue are \emph{concatenation} and
\emph{majority voting}. In the former, one concatenates the sequences originating 
from several genes and hopes that a tree inferred from the combined data will produce 
a better estimate. 
This approach appears to give poor results~\cite{KubatkoDegnan:07}. Alternatively,
one can infer multiple gene trees and output the most common reconstruction
(that is, take a majority vote). This is also often doomed to failure. 
Indeed, a recent, striking result of Degnan and Rosenberg~\cite{DegnanRosenberg:06}
shows that, under appropriate conditions, the \emph{most likely} gene tree may be
inconsistent with the species tree; and this situation may arise on \emph{any} topology
with at least $5$ species. See also~\cite{PamiloNei:88,Takahata:89} for related
results.

Other techniques are being explored that attempt to address incomplete lineage sorting,
notably Bayesian~\cite{EdLiPe:07} and likelihood~\cite{SteelRodrigo:07} methods.
However the problem is still far from being solved as discussed 
in~\cite{MaddisonKnowles:06}. Here we propose
a simple technique---which we call Global LAteSt Split or GLASS---for estimating species trees 
from multiple genes (or \emph{loci}). Our technique 
develops some of the ideas of Takahata~\cite{Takahata:89} and Rosenberg~\cite{Rosenberg:02} 
who studied the properties of gene trees in terms of the corresponding 
species tree.
In our main result, we show that GLASS is \emph{statistically consistent}, that is, 
it always returns the correct topology given sufficiently many (unlinked) genes---thereby 
avoiding the pitfalls highlighted in~\cite{DegnanRosenberg:06}. We also obtain explicit 
convergence rates under a standard model based on Kingman's coalescent~\cite{Kingman:82}.
Moreover, we allow the use of several alleles from each population and we show how
our technique leads to an extension of Rosenberg's \emph{topological concordance}~\cite{Rosenberg:02}
to multiple loci.

We note the recent results of Steel and Rodrigo~\cite{SteelRodrigo:07} who
showed that Maximum Likelihood (ML) is statistically consistent under slightly different assumptions.
An advantage of GLASS over likelihood (and Bayesian) methods 
is its computational efficiency, as no efficient algorithm for finding ML trees is known. 
Furthermore,  GLASS gives explicit convergence rates---useful in assessing the quality of the reconstruction.

For more background on phylogenetic inference and coalescent theory, 
see e.g.~\cite{Felsenstein:04, SempleSteel:03, HeScWi:05, Nordborg:01, Tavare:04}.

\paragraph{Organization.} The rest of the paper is organized as follows. We begin in
Section~\ref{section:setup} with a description of the basic setup. The GLASS method
is introduced in Section~\ref{section:glass}. A proof of its consistency can be found
in Sections~\ref{section:combinatorial} and~\ref{section:consistency}. 
We show
in Section~\ref{section:noisy}
that GLASS remains consistent under moderate estimation errors.
Finally in Section~\ref{section:generalization} 
we do away with the molecular
clock assumption and we show how our technique can be used in conjunction 
with any distance matrix method.

\section{Basic Setup}\label{section:setup}

We introduce our basic modelling assumptions. See e.g.~\cite{DegnanRosenberg:06}.

\paragraph{Species tree.} Consider $n$ isolated populations with a common
evolutionary history given by the \emph{species tree} $S = (V,E)$ 
with leaf set $L$. Note that $|L| = n$. For each branch $e$ of $S$,
we denote: 
\begin{itemize}
\item $N_e$, the (haploid) population size on $e$ 
(we assume that the population size remains constant along the branch);

\item $t_e$, the number of generations encountered on $e$;

\item $\tau_e = \frac{t_e}{2N_e}$, the length of $e$ in standard
coalescent time units;

\item $\mincot = \min_e \tau_e$, the shortest branch length in $S$.
\end{itemize}
The model does not allow migration between contemporaneous populations.
Often in the literature, the population sizes $\{N_e\}_{e\in E}$, are taken to be equal
to a constant $N$. Our results are valid in a more general setting.   
%A species tree on $n$ populations experiences $n-1$ divergences.
%We denote $t_1\leq t_2\leq\cdots\leq t_{n-1}$ the divergence times in number of
%generations going backwards in time.

\paragraph{Gene trees.}
We consider $k$ \emph{loci} $\loci$.
For each population $l$ and each locus $i$, we sample a set of alleles $\alleles{i}{l}$.
Each locus $i\in \loci$ has a genealogical history represented by a 
\emph{gene tree} $\gcal^{(i)} = (\vcal^{(i)}, \ecal^{(i)})$
with leaf set $\lcal^{(i)} = \cup_l \alleles{i}{l}$.  
%For each branch $e$ in $\gcal^{(i)}$ we denote by $\lambda^{(i)}_e$ the number of generations
%along $e$.
For two leaves $a,b$ in $\gcal^{(i)}$, we let $\divtime{i}{ab}$ be
the time in number of generations to the most recent common ancestor of $a$ and $b$ in $\gcal^{(i)}$.
%, that is,
%\begin{equation*}
%d^{(i)}_{ab} = \frac{1}{2}\sum_{e\in\mathrm{P}^{(i)}(a,b)} \lambda^{(i)}_e,
%\end{equation*}
%where $\mathrm{P}^{(i)}(a,b)$ is the path (set of edges) between $a$ and $b$ in 
%gene tree $\gcal^{(i)}$. 
Following~\cite{Takahata:89, Rosenberg:02}
we are actually interested in \emph{interspecific} coalescence times. Hence,
we define, for all $r,s \in L$,
\begin{equation*}
\divtime{i}{rs} = \min\left\{
\divtime{i}{ab} \ :\ a\in \alleles{i}{r}, b\in \alleles{i}{s}
\right\}.
\end{equation*}
  
\paragraph{Inference problem.} We seek to solve the following inference problem.
We are given $k$ gene trees as above, including accurate estimates
of the coalescence times 
\begin{equation*}
\left\{\left(\divtime{i}{ab}\right)_{a,b\in \lcal^{(i)}}\right\}_{i\in \loci}. 
\end{equation*}
Our goal is to infer the species tree $S$. 

\paragraph{Stochastic Model.}
In Section~\ref{section:combinatorial}, we will first state 
the correcteness of our inference algorithm in terms 
of a combinatorial property of the gene trees. 
In Section~\ref{section:consistency}, we will then show that 
under the following standard stochastic assumptions, this property 
holds for a moderate number of genes. 

Namely, we will assume that each gene tree $\gcal^{(i)}$ is distributed
according to a standard \emph{coalescent process}: looking backwards in time, 
in each branch any two alleles coalesce at exponential rate 1 independently of 
all other pairs; whenever two populations merge in the species tree, 
we also merge the allele sets of the corresponding populations
(that is, the coalescence proceeds on the \emph{union} of both allele sets). 
We further assume that the $k$ loci $\loci$ are \emph{unlinked}
or in other words that the gene trees $\{\gcal^{(i)}\}_{i\in \loci}$ are mutually
independent. 

Under these assumptions, an inference algorithm
is said to be \emph{statistically consistent} if the probability
of returning an incorrect reconstruction goes to 0 as $k$ tends to
$+\infty$.

\section{Species Tree Estimation}\label{section:glass}

We introduce a technique which we call the
Global LAteSt Split (GLASS) method.

\paragraph{Inference method.}
Consider first the case of a single gene ($k=1$).
Looking backwards in time, the first speciation occurs at some time
$T_1$, say between populations $r_1$ and $s_1$. It is well known that,
for any sample $a$ from $\alleles{1}{r_1}$ and $b$ from $\alleles{1}{s_1}$, the coalescence time
$\divtime{1}{ab}$ between alleles $a$ and $b$ \emph{overestimates}
the divergence time of the populations. As noted in~\cite{Takahata:89}, 
a better estimate of $T_1$ can be obtained
by taking the smallest interspecific coalescence time between
alleles in $\alleles{1}{r_1}$ and in $\alleles{1}{s_1}$, that is, by considering instead
$\divtime{1}{r_1 s_1}$. 

The inference then proceeds as follows. 
First, cluster the two populations, say $r_1$ and $s_1$,
with smallest interspecific coalescence time $\divtime{1}{r_1 s_1}$. 
Define the coalescence time of two clusters $A, B \subseteq L$ as the
minimum interspecific coalescence time between populations in $A$ and in $B$, that is,
\begin{equation*}
\divtime{1}{AB} = \min\left\{
\divtime{1}{rs}\ :\ r\in A, s\in B
\right\}.
\end{equation*}
Then, repeat as above until there is only one cluster left.
This is essentially the algorithm proposed by Rosenberg~\cite{Rosenberg:02}.
In particular, Rosenberg calls the implied topology on the populations so obtained
the \emph{collapsed gene tree}.

How to extend this algorithm to $k>1$? 
As we discussed earlier, one could infer a gene tree as above for each locus and
take a majority vote---but this approach fails~\cite{DegnanRosenberg:06};
in particular, it is generally not statistically consistent.

Another natural idea is to get a ``better'' estimate of coalescence times by \emph{averaging}
across loci. This leads to the 
Shallowest Divergence Clustering method of Maddison and Knowles~\cite{MaddisonKnowles:06}.
We argue that a better choice is, instead, to take the \emph{minimum} across loci.
In other words, we apply the clustering algorithm above to the quantity
\begin{equation*}
\divt{AB} = \min\left\{
\divtime{i}{AB}\ :\ i\in \loci
\right\},
\end{equation*}
for all $A, B \subseteq L$ with $A\cap B = \emptyset$. 
The reason we consider the minimum is similar to the case
of one locus and several samples per population above: it suffices to have
\emph{one} pair $a \in \alleles{i}{r}$, $b \in \alleles{i}{s}$ (for some $i$) with coalescence time
$T$ across \emph{all pairs of samples in populations $r$ and $s$ (one from each)} 
and \emph{all loci in $\loci$} to provide
indisputable evidence that the corresponding species branch
before time $T$ (looking backwards in time). In a sense, 
we build the ``minimal'' tree on $L$ that is ``consistent'' with
the evidence provided by the gene trees.
This type of approach is briefly discussed by Takahata~\cite{Takahata:89} in the simple case of three
populations (where the issues raised by~\cite{DegnanRosenberg:06} do not arise). 

The algorithm, which we name GLASS, is detailed in Figure~\ref{figure:glass}.
We call the tree so obtained the \emph{glass tree}.
We show in the next section that GLASS is in fact statistically consistent.
\begin{figure*}[h]
\framebox{
\begin{minipage}{12.2cm}
{\small \textbf{Algorithm} GLASS\\
\textit{Input:} Gene trees $\{\gcal^{(i)}\}_{i\in \loci}$ and coalescence times
$\divtime{i}{ab}$ for all $i\in \loci$ and $a,b\in \lcal^{(i)}$;\\
\textit{Output:} Estimated topology $S'$;
\begin{itemize}
\item \itemname{Intercluster\ coalescences} For all $A, B \subseteq L$ with
$A\cap B = \emptyset$, compute
\begin{equation*}
\divt{AB} = \min\left\{
\divtime{i}{ab}\ :\ i\in \loci, r\in A, s\in B, a\in \alleles{i}{r}, b\in \alleles{i}{s}
\right\};
\end{equation*}
\item \itemname{Clustering} Set 
$Q := \{\{r\}\ :\ r\in L\}$;
Until $|Q| = 1$:
\begin{itemize}
\item Denote the current partition $Q = \{A_1,\ldots,A_z\}$;
\item Let $A', A''$ minimize $\divt{AB}$ over all pairs $A,B \in Q$ (break ties arbitrarily);
\item Merge $A'$ and $A''$ in $Q$;
\end{itemize}
\item \itemname{Output} Return the topology implied by the steps above.
\end{itemize}
}
\end{minipage}
} \caption{Algorithm GLASS.} \label{figure:glass}
\end{figure*}

\paragraph{Multilocus concordance.} A gene tree with one sample per population
is said to be \emph{concordant} (sometimes also ``congruent'' or ``consistent'')
with a species tree if their (leaf-labelled) topologies agree. When the number
of samples per population is larger than one, one cannot directly compare
the topology of the gene tree with that of the species tree since they contain a different
number of leaves. Instead, Rosenberg~\cite{Rosenberg:02} defines a gene tree to be 
\emph{topologically concordant} with a species tree if the \emph{collapsed gene tree}
(see above) coincides with the species tree. 

We extend Rosenberg's definition to multiple loci. We say that a collection
of gene trees $\{\gcal^{(i)}\}_{i\in \loci}$ is \emph{multilocus concordant} 
with a species tree $S$ if the \emph{glass tree} agrees with the species tree.
Therefore, to prove that GLASS is statistically consistent, it suffices to show that
the probability of multilocus concordance goes to $1$ as the number of loci
goes to $+\infty$.

\section{Sufficient Conditions}\label{section:combinatorial}

In this section, we state a simple combinatorial condition
guaranteeing that GLASS returns the correct species tree. Our condition is
an extension of Takahata's condition in the case of a single gene~\cite{Takahata:89}.
See also~\cite{Rosenberg:02}.

As before, let $S$ be a species tree and $\{\gcal^{(i)}\}_{i\in I}$ a collection
of gene trees. For a subset of leaves $A \subseteq L$, denote by $\mrca{A}$ the most
recent common ancestor (MRCA) of $A$ in $S$. 
For a (internal or leaf) node $v$ in $S$, we use the following notation:
\begin{itemize}
\item $\desc{v}$ are the descendants of $v$ in $L$; 
\item $\undert{v}$ is the time elapsed in number of generations
between $v$ and $\desc{v}$; 
\item $\overt{v}$ is the time elapsed in number of generations
between the immediate ancestor of $v$ and $\desc{v}$.
\end{itemize}
In particular, note that if $e$ is the branch immediately above $v$, then we have
\begin{equation*}
t_e = \overt{v} - \undert{v}.
\end{equation*}
Also, we call the subtree below $v$, \emph{clade} $v$.

Our combinatorial condition can be stated as follows: 
\begin{equation*}
(\star) \quad \forall u,v \in V,\ \undert{\mrca{\desc{u}\cup \desc{v}}} 
\leq \divt{\desc{u} \desc{v}} 
< \overt{\mrca{\desc{u}\cup \desc{v}}}.
\end{equation*}  
In words, for any two clades $u$, $v$, there is at least one locus $i$ and one pair
of alleles $a, b$ with $a$ from clade $u$ and $b$ from clade $v$ 
such that the lineages of $a$ and $b$ coalesce before the end of the branch above the MRCA of $u$ and $v$.
(The first inequality is clear by construction.)
By the next proposition, condition $(\star)$ is sufficient for multilocus concordance. 
Note, however, that it is not necessary. Nevertheless note that, by design, GLASS always returns
a tree, even when the condition is not satisfied.
\begin{proposition}[Sufficient Condition]\label{thm:sufficient}
Assume that $(\star)$ is satisfied. Then, GLASS returns the correct species tree.
In other words, the gene trees $\{\gcal^{(i)}\}_{i\in \loci}$ are multilocus concordant
with the species tree $S$.
\end{proposition}
\begin{proof}
Let $Q$ be one of the partitions obtained by GLASS along its execution and let 
$B$ be the newly created set in $Q$. We claim that, under $(\star)$, it must be the case that
\begin{equation}\label{eq:alpha}
B = \desc{\mrca{B}}.
\end{equation} 
That is, $B$ is the set of leaves of a clade in the species tree $S$. 
The proposition follows immediately from this
claim.

We prove the claim by induction on the execution time of the algorithm. 
Property (\ref{eq:alpha}) is trivially true initially. 
Assume the claim holds up to time $T$ and let $Q$, as above, be the partition at time $T+1$. 
Note that $B$ is obtained
by merging two sets $B'$ and $B''$ forming a partition of $B$. By induction,
$B'$ and $B''$ satisfy (\ref{eq:alpha}). 
Now, suppose by contradiction that $B$
does not satisfy (\ref{eq:alpha}). 
Let $\leftd{\mrca{B}}$ and $\rightd{\mrca{B}}$ be the clades immediately below $\mrca{B}$ 
with corresponding
leaf sets $C' = \desc{\leftd{\mrca{B}}}$ and $C'' = \desc{\rightd{\mrca{B}}}$. 
By our induction hypothesis, each of $B'$ and $B''$
must be contained in one of $C'$ or $C''$. Say $B' \subseteq C'$ and
$B'' \subseteq C''$ without loss of generality. 
Moreover, since $B$ does not satisfy (\ref{eq:alpha}),
one of the inclusions is strict, say $B' \subset C'$. But by $(\star)$, any set
$X$ in $Q$ containing an element of $C'-B'$ has
\begin{equation}\label{eq:dbx}
\divt{B'X} 
< \overt{\mrca{B' \cup X}}
\leq \overt{\leftd{\mrca{B}}} 
= \undert{\mrca{B}} 
= \undert{\mrca{B'\cup B''}} 
\leq \divt{B'B''}.
\end{equation}
To justify the first two inequalities above, 
note that $X$ is contained in the partition at time $T$ and therefore
satisfies (\ref{eq:alpha}). In particular, by construction
\begin{equation*}
B' \cup X \subseteq C'.
\end{equation*}
Hence by (\ref{eq:dbx}), 
GLASS would not have merged $B'$ and $B''$, a contradiction. 
\end{proof}

\section{Statistical Consistency}\label{section:consistency}

In this section, we prove the consistency of GLASS.

\paragraph{Consistency.}
We prove the following consistency result. Note that the theorem holds
for any species tree---including the ``anomaly zone'' of Degnan and Rosenberg~\cite{DegnanRosenberg:06}.
\begin{proposition}[Consistency]\label{thm:consistency}
GLASS is statistically
consistent.
\end{proposition}
\begin{proof}
Throughout the proof, time runs \emph{backwards} as is conventional in coalescent theory.
We use Proposition~\ref{thm:sufficient} and give a lower bound on the probability
that condition $(\star)$ is satisfied.

Consider first the case of one locus and one sample per population. 
By $(\star)$, the reconstruction
is correct if every time two populations meet, 
the corresponding alleles coalesce before the end of the branch immediately above.
By classical coalescent calculations (e.g.~\cite{Tavare:84}), this happens with probability
at least
\begin{equation*}
(1 - e^{-\mincot})^{n-1},
\end{equation*}
where we used the fact that there are $n-1$ divergences. 

Now consider the general case. 
Imagine running the coalescent processes of all loci \emph{simultaneously}. 
Consider any branching between
two populations. In every gene tree separately, if several alleles emerge on either sides of the branching, 
choose arbitrarily one allele from each side. 
The probability that the chosen allele pairs fail to coalesce
before the end of the branch above in \emph{all} loci is at most $e^{-k\mincot}$
by independence. Indeed, irrespective of everything else going on, 
two alleles meet at exponential rate 1 (conditionally on the past).
This finally gives a probability of success of at least
\begin{equation*}
(1 - e^{-k\mincot})^{n-1}.
\end{equation*}

For $n$ and $\mincot$ fixed, we get
\begin{equation*}
(1 - e^{-k\mincot})^{n-1} \to 1,
\end{equation*} 
as $k \to +\infty$, as desired.
\end{proof}

\paragraph{Rates.} Implicit in the proof of Proposition~\ref{thm:consistency} is the following
convergence rate.
\begin{proposition}[Rate]\label{thm:rate}
It holds that
\begin{equation*}
\prob[\mathrm{\ Multilocus\ Discordance\ }] \leq (n-1)[e^{-\mincot}]^k.
\end{equation*}
In particular, for any $\eps > 0$, taking
\begin{equation*}
k = \frac{1}{\mincot} \ln\left(\frac{n-1}{\eps}\right),
\end{equation*} 
we get
\begin{equation*}
\prob[\mathrm{\ Multilocus\ Discordance\ }] \leq \eps.
\end{equation*}
\end{proposition}
\begin{proof}
Note that
\begin{equation*}
1 - (1 - e^{-k\mincot})^{n-1} \leq (n-1)[e^{-\mincot}]^k.
\end{equation*}
\end{proof}

\paragraph{Multiple alleles v. multiple loci.}
It is interesting to compare the relative effects of adding more alleles or more
loci on the accuracy of the reconstruction. 
The result in Proposition~\ref{thm:rate} does not address this question.
In fact, it is hard to obtain useful analytic expressions 
for small numbers of genes and alleles. However, the
asymptotic behavior is quite clear. 
Indeed, as was pointed out in~\cite{Rosenberg:02} (see 
also~\cite{MaddisonKnowles:06} for empirical evidence), 
the benefit of adding more alleles eventually wears out. 
This is because the probability of observing any given number
of alleles at the top of a branch is uniformly bounded
in the number alleles existing at the bottom.
More precisely, we have the following result
which is to be contrasted with Proposition~\ref{thm:rate}.
\begin{proposition}[Multiple Alleles: Saturation Effect]
Let $S$ be any species tree on $n$ populations.
Then, there is a $0 < q^* < 1$ (depending only on $S$) such that 
for any number of loci $k > 0$ and
any number of alleles sampled per population, we have
\begin{equation*}
\prob[\mathrm{\ Multilocus\ Discordance\ }] \geq (q^*)^k > 0.
\end{equation*}
In particular, for a fixed number of loci $k > 0$, 
as the number of alleles per populaiton goes to $+\infty$, 
the probability that GLASS correctly reconstructs
$S$ remains bounded away from $1$.
\end{proposition}
\begin{proof}
Take any three populations $a, b, c$ from $S$. Assume that $a$ and $b$ meet
$T_1$ generations back and that $c$ joins them $T_2$ generations later.
For $w = a,b,c$ and $i\in \loci$, 
let $Y^{(i)}_w$ be the event that in locus $i$ there is only one allele
remaining at the top of the branch immediately above $w$. Let $Z^{(i)}$
be the event that the topology of gene tree $i$ restricted to $\{a,b,c\}$
is topologically discordant with $S$. It follows from bound (6.5) in~\cite{Tavare:84}
that there is $0 < q' < 1$ independent of $h$ such that
\begin{equation*}
\prob[Y^{(i)}_{w}] \geq q',
\end{equation*}
for all $i \in \loci$ and $w \in \{a,b,c\}$. Also, it is clear that 
there is $0 < q'' < 1$ depending on $T_2$
such that
\begin{equation*}
\prob[Z^{(i)}\,|\,Y^{(i)}_w,\ \forall w\in \{a,b,c\}] \geq q'',
\end{equation*} 
for all $i\in \loci$. Therefore, by independence of the loci,
we have
\begin{eqnarray*}
&&\prob[\mathrm{\ Multilocus\ Discordance\ }]\\
&&\qquad\qquad\qquad\qquad \geq 
\prod_{i\in \loci}\prob[Z^{(i)}\,|\,Y^{(i)}_w,\ \forall w\in \{a,b,c\}]\prod_{w\in\{a,b,c\}}\prob[Y^{(i)}_{w}]\\
&&\qquad\qquad\qquad\qquad \geq ((q')^3 q'')^k.
\end{eqnarray*}
Take $q^* = (q')^3 q''$. That concludes the proof.
\end{proof}

\section{Tolerance to Estimation Error}\label{section:noisy}

The results of the previous section are somewhat unrealistic in that they assume
that GLASS is given \emph{exact} estimates of coalescence times. In this
section, we relax this assumption. 

Assume that the input to the algorithm is now a set of \emph{estimated} coalescence
times 
\begin{equation*}
\left\{\left(\estdivtime{i}{ab}\right)_{a,b\in \lcal^{(i)}}\right\}_{i\in \loci},
\end{equation*}
and, for all $A, B \subseteq L$, let 
\begin{equation*}
\estdivt{AB} = \min\left\{
\estdivtime{i}{ab}\ :\ i\in \loci, r\in A, s\in B, a\in \alleles{i}{r}, 
b\in \alleles{i}{s}
\right\},
\end{equation*}
be the corresponding estimated intercluster coalescence times computed
by GLASS. Assume further that there is a $\delta > 0$ such that
\begin{equation*}
\left|\estdivtime{i}{ab} - \divtime{i}{ab}\right| \leq \delta,
\end{equation*}
for all $i\in \loci$ and $a,b\in \lcal^{(i)}$. In particular, note that
\begin{equation*}
\left|\estdivt{AB} - \divt{AB}\right| \leq \delta,
\end{equation*} 
for all $A, B \subseteq L$.

Let $\mint$ be the shortest branch length in number of generations, that is,
\begin{equation*}
\mint = \min\{t_e\ :\ e \in E\}.
\end{equation*}
%Denote $\widehat \mint = \mint - 2\delta$. 
We extend our combinatorial condition
$(\star)$ to
\begin{equation*}
(\hat\star) \quad \forall u,v \in V,\ 
\undert{\mrca{\desc{u}\cup \desc{v}}} 
\leq \divt{\desc{u} \desc{v}} 
< \overt{\mrca{\desc{u}\cup \desc{v}}} - 2\delta.
\end{equation*}  
Then, we get the following.
\begin{proposition}[Sufficient Condition:Noisy Case]\label{thm:sufficient-noisy}
Assume that
\begin{equation}\label{eq:m3}
\delta < \frac{\mint}{2},
\end{equation}
and that $(\hat\star)$ is satisfied. Then, GLASS returns the correct species tree.
\end{proposition}
\begin{proof}
The proof follows immediately from the argument in Proposition~\ref{thm:sufficient}
by noting that equation (\ref{eq:dbx}) becomes
\begin{eqnarray*}
\estdivt{B'X} 
&\leq& \divt{B'X} + \delta\\
&<& \overt{\mrca{B'\cup X}} - \delta\\
&\leq& \undert{\mrca{B' \cup B''}} - \delta\\
&\leq& \divt{B' B''} - \delta\\
&\leq& \estdivt{B'B''}.
\end{eqnarray*}
Condition (\ref{eq:m3}) ensures that $(\hat\star)$
is satisfiable.
\end{proof}
Moreover, we have immediately:
\begin{proposition}[Consistency \& Rate: Noisy Case]\label{thm:consistency-noisy}
Assume that
\begin{equation*}
\delta < \frac{\mint}{2}.
\end{equation*}
Then GLASS is statistically consistent. Moreover, let
\begin{equation*}
\consrate = \frac{m - 2\delta}{m},
\end{equation*}
then it holds that
\begin{equation*}
\prob[\mathrm{\ Incorrect\ Reconstruction\ }] \leq (n-1)[e^{-\mincot\consrate}]^k.
\end{equation*}
In particular, for any $\eps > 0$, taking
\begin{equation*}
k = \frac{1}{\mincot\consrate} \ln\left(\frac{n-1}{\eps}\right),
\end{equation*} 
we get
\begin{equation*}
\prob[\mathrm{\ Incorrect\ Reconstruction\ }] \leq \eps.
\end{equation*}
\end{proposition}

\section{Generalization}\label{section:generalization}

The basic observation underlying our approach is that distances between
populations may be estimated correctly using the minimum divergence time
among all individuals and all genes.

Actually, this observation may be used in conjunction with any distance-based
reconstruction algorithm. (See e.g.~\cite{Felsenstein:04,SempleSteel:03}
for background on distance matrix methods.) 
This can be done under very general assumptions as
we discuss next. First, we do away with the molecular clock assumption.
Indeed, it turns out that $\divtime{i}{ab}$ need not be the
divergence time between $a$ and $b$ for gene $i$.
Instead, we take $\divtime{i}{ab}$ to be the molecular distance 
between $a$ and $b$ in gene $i$, that is,
the time elapsed from the divergence point 
to $a$ and $b$ integrated against the rate of mutation.
We require that the rate of mutation be the same for all genes and all
individuals in the same branch of the species tree, but we allow rates to 
differ across branches. 
Below, all quantities of the type
$\divt{},\estdivt{}{}$ etc.~are given in terms of this
molecular distance. 

For any two clusters $A,B \subseteq L$, we define
\begin{equation}\label{eq:newest}
\estdivt{AB} = \min\left\{
\estdivtime{i}{ab}\ :\ i\in \loci, r\in A, s\in B, a\in \alleles{i}{r},
b\in \alleles{i}{s}
\right\},
\end{equation}
as before. 
Let
\begin{equation*}
\mint' = \min\{t_e\rho_e\ :\ e \in E\},
\end{equation*}
where $\rho_e$ is the rate of mutation on branch $e$.
%Note that under conditions of the form
%\begin{equation*}
%\forall u,v \in L,\
%\undert{\mrca{\desc{u}\cup \desc{v}}}
%\leq \divt{\desc{u} \desc{v}}
%< \undert{\mrca{\desc{u}\cup \desc{v}}} + \frac{\mint'}{8},
%\end{equation*}
%and
%\begin{equation*}
%\delta \leq \frac{\mint'}{8},
%\end{equation*}
It is easy to generalize condition $(\hat\star)$
so that we can use (\ref{eq:newest}) to estimate all molecular distances
between pairs of populations up to an additive error of, say, $\mint'/4$. 
Then using
standard four-point methods, 
we can reconstruct the species tree correctly.

Note furthermore that by the results of~\cite{ErStSzWa:99a}, 
it suffices in fact to estimate distances between pairs of populations 
that are ``sufficiently close.'' We can derive 
consistency conditions which guarantee the reconstruction of the correct species tree
in that case as well.

\bibliographystyle{alpha}
\bibliography{thesis}

\end{document}